\documentclass{article}
\usepackage{amsmath,amssymb,textcomp}
\usepackage{graphicx,mathrsfs}
\usepackage{multirow}
\usepackage{xcolor}

\newtheorem{theorem}{Theorem}

\newtheorem{definition}[theorem]{Definition}
\newtheorem{example}[theorem]{Example}

\newtheorem{proposition}[theorem]{Proposition}

\newenvironment{proof}[1][Proof]{\textbf{#1.} }{\ \rule{0.5em}{0.5em}}

\begin{document}

\title{From the Heisenberg to the Schr\"{o}dinger Picture:\\Quantum Stochastic Processes and Process Tensors\thanks{To appear in Proceedings of the 60th IEEE Conference on Decision and Control (CDC), Dec. 13-15, 2021}}

\author{Hendra I. Nurdin\thanks{H. I. Nurdin is with the School of Electrical Engineering and 
Telecommunications,  UNSW Australia,  Sydney NSW 2052, Australia (\texttt{email: h.nurdin@unsw.edu.au})} \and John Gough  \thanks{J. Gough is with the Department of Physics, Aberystwyth University, Ceredigion, SY23 3BZ, Wales, UK (\texttt{email: jug@aber.ac.uk}). JG acknowledges funding under ANR grant (ANR-19-CE48-0003)} }

\date{}

\maketitle

\begin{abstract}
A general theory of quantum stochastic processes was formulated by Accardi, Frigerio and Lewis in 1982 within the operator-algebraic framework of quantum probability theory, as a non-commutative extension of the Kolmogorovian classical stochastic processes. More recently, studies on non-Markovian quantum processes have led to the  discrete-time process tensor formalism in the Schr\"{o}dinger picture to describe the outcomes of sequential interventions on open quantum systems. However, there has been no treatment of the relationship of the process tensor formalism to the quantum probabilistic theory of quantum stochastic processes. This paper gives an exposition of quantum stochastic processes and the process tensor and the relationship between them. In particular, it is shown how the latter emerges from the former via extended correlation kernels incorporating ancillas.  
\end{abstract}

\section{Introduction}
\label{sec:intro}
Modern probability theory as formulated by Kolmogorov \cite{Kolm33} underpins the theory of stochastic processes, stochastic systems and stochastic control \cite{WIlliams91,WH85}. Similarly, beginning with the seminal work of von Neumann on the axiomatization of quantum mechanics \cite{VN18}, quantum probability theory has emerged as a non-commutative generalization of probability theory \cite{Streat00}. It provides a natural setting for a theory of quantum stochastic processes as a non-commutative generalization of the classical theory of Kolmogorov. A major departure of the quantum setting from the classical one is that the random outcomes of sequential measurements on quantum stochastic processes do not in general satisfy the Kolmogorov consistency conditions and hence cannot be described as classical stochastic processes with well-defined sample paths. 

A general quantum probabilistic theory of quantum stochastic processes  was introduced in the seminal work of Accardi, Frigerio and Lewis (AFL) \cite{AFL82}. The formulation is given in the  Heisenberg picture, generalizing the Kolmogorovian theory of classical stochastic processes. This is in the sense that observables as quantum random variables evolve with time while the state of the system is kept fixed, just as how random variables evolve in time in a classical stochastic process while the probability measure on the underlying classical probability space remains fixed. Quantum stochastic processes as operator-valued processes are defined independently of  single or sequential measurements that may be performed on the process at any time. However, measurements and their probabilistic outcomes are accounted for by the correlation kernels of quantum stochastic processes, playing  a similar role to the family of finite-dimensional distributions for classical stochastic processes. 


Recent efforts in trying to understand and characterize temporal quantum noise in engineered quantum systems have led to the consideration of alternative formalisms for defining, describing and witnessing non-Markovian quantum processes, see, e.g., \cite{RHP14}. Unlike \cite{AFL82}, these formalisms are developed in the Schr\"{o}dinger picture, with an emphasis on the transformations of states (described by a density operator) of the system of interest. The process tensor formalism was introduced in \cite{PRRFPM18} to overcome the limitations of conventional descriptions of non-Markovian dynamics based on  the reduced dynamics of the system in the Schr\"{o}dinger picture (as linear transformations of the system's states). This includes addressing initial system-environment correlation  and multi-time interventions.

However,  there have  been so far no studies to reconcile  the process tensor formalism to the well-established AFL theory and its subsequent developments. Since both formalisms are concerned with related objects and the same physics but set in different pictures (Heisenberg vs Schr\"{o}dinger), one would expect a close relationship between the two. In this paper, we connect the process tensor to AFL theory through the correlation kernels of quantum stochastic processes. In particular, it is shown how process tensors can be recovered from extended correlation kernels incorporating ancillas. Along the way, we also give a tutorial style overview on quantum stochastic processes, multi-time correlations and sequential measurements on quantum systems. In particular, we highlight subtle points surrounding multi-time correlations and sequential measurements, emphasizing the latter's departure from Kolmogorovian classical stochastic processes.


\noindent \textbf{Notation.} $\mathbb{R}$ and $\mathbb{C}$ denote the real and complex numbers, respectively. For $c \in \mathbb{C}$, $\overline{c}$ is its complex conjugate. For a complex-valued function $X$, $\overline{X}(\cdot) = \overline{X(\cdot)}$. For a set $S$, $S^n$ denotes the $n$-fold direct product $S^n =\underbrace{S \times S \times \cdots S}_{\hbox{$n$ times}}$.  The notation $\otimes$ denotes the tensor product of Hilbert spaces and the algebraic tensor product of linear operators. A Dirac ket  $|x\rangle$ denotes a complex  vector in a Hilbert space  while a bra $\langle x |$  denotes the conjugate transpose (or dual functional) of the  vector. Thus $\langle x|y \rangle$ is the inner product of $|x\rangle$ and $|y \rangle$.  For any operator $X$ mapping a Hilbert space to another, $X^{*}$ denotes the adjoint of $X$ and  $\mathrm{tr}(X)$ denotes the trace of a trace-class operator.  $X^{\top}$ denotes the transpose of a matrix $X$. $\mathrm{B}(\mathfrak{h})$ and $\mathcal{S}(\mathfrak{h})$  denote the complex space of all bounded operators and the  convex cone of all unnormalised density operators over a Hilbert space $\mathfrak{h}$, respectively. For a set of distinct  numbers $t_1, t_2,\cdots,t_n \in \mathbb{R}$, a time tuple is the  $n$-tuple $\mathbf{t}_n=(t_1,t_2,\ldots,t_n)$. Non-strict set inclusion is denoted by $\subseteq$ while strict inclusion is denoted by $\subset$.  The composition operation is denoted by $\circ$.

\section{Quantum stochastic processes}
\label{sec:qsp}

A quantum probability space is a pair $(\mathscr{X},\mu)$ where $\mathscr{X}$ is a von Neumann algebra of bounded operators on some Hilbert space $\mathfrak{h}$ (containing the identity operator $I_{\mathscr{X}}$) and and $\mu$ is a unital normal state on $\mathscr{X}$ (unital meaning $\mu(I_{\mathscr{X}})=1$). Recall that a von Neumann algebra is a *-algebra of operators equipped with addition, multiplication (as composition of operators) and involution $^*$ (defined as the adjoint of the operator) and is closed with respect to the normal topology of sub-algebras of $\mathrm{B}(\mathfrak{h})$. For normal states $\mu$, there exists a density operator $\rho$ on $\mathfrak{h}$ such that $\mu(X) = {\rm tr}(\rho X)$ for all $X \in \mathscr{X}$; see \cite{BvHJ07} and the references therein. We will also write $\mu(\cdot)$ using the quantum expectation notation $\langle \cdot \rangle$. 

A classical probability space $(\Omega,\mathcal{F},P)$ can be viewed as a Banach algebra $L^{\infty}(\Omega,\mathcal{F},P)$ of essentially bounded random variables on $(\Omega,\mathcal{F},P)$. A  quantum probability space $(\mathscr{C},\mu)$, with $\mathscr{C}$ commutative\footnote{Meaning that the elements of $\mathcal{C}$ are commuting with one another.} and $\mu$ a unital normal state, is *-isomorphic to a classical probability space $(L^{\infty}(\Omega,\mathcal{F},\nu),\mathbb{E})$, where the expectation operator $\mathbb{E}(X)=\int_{\Omega} X(\omega) P(d\omega)$ for some measure $P$ that is absolutely continues with respect to $\nu$.  This *-isormorphism is a bijective map $\iota: (\mathscr{C},\mu) \rightarrow (L^{\infty}(\Omega,\mathcal{F},\nu),\mathbb{E})$ with the properties $\iota(ab) = \iota(a) \iota(b)$ and $\iota(b^*) = \overline{\iota(b)}$ for any $a,b \in (\mathscr{C},\mu)$. The *-isomorphism be tween a classical probability space and commutative von Neumann algebra is known as the {\em Spectral Theorem}, see, e.g., \cite[Theorem 3.3]{BvHJ07}.


The physical interpretation of the quantum probability space $(\mathscr{X},\mu)$ is as follows. The underlying Hilbert space of the operators in  $\mathscr{X}$ is the Hilbert space of an associated quantum mechanical system. Observables of the system are  the self-adjoint operators in $\mathscr{X}$, which are also quantum random variables (i.e., quantum analogues of real-valued random variables). The quantum expectation of an observable $X$ is given by $\mu(X)$.  Events $E \in \mathscr{X}$ are projection operators  $(E=E^*=E^2)$ and the probability of an event $E$ is given by $\mu(E)$. Only commuting events $E$ can have a joint probability distribution that satisfy the Kolmogorov consistency conditions, non-commuting events cannot be assigned a joint probability distribution. This can be seen as a direct consequence of the Spectral Theorem: since non-commuting events form the elements of a non-commutative algebra it cannot be mapped to a classical probability space.

Let $T \subseteq \mathbb{R}$. A quantum stochastic process over a von Neumann algebra $\mathscr{B} \subseteq \mathrm{B}(\mathfrak{h})$ is a triplet $(\mathscr{A},\{j_t\}_{t \in T},\mu)$, with $\mathscr{A}$ another von Neumann algebra of operators, possibly over another Hilbert space $\mathfrak{k}$, $\mu$ a normal state on $\mathscr{A}$, and $j_t: \mathscr{B} \rightarrow \mathscr{A}$ $\forall t \in T$ is a *-homomorphism  from $\mathscr{B}$ to $\mathscr{A}$, $j_t(XY)=j_t(X)j_t(Y)$ and $j_t(X^*)=j_t(X)^*$ for any $X,Y \in \mathscr{B}$.   Note that $(\mathscr{A},\mu)$ is a quantum probability space.

Since  a collection of non-commuting random variables will not have a joint probability distribution, for quantum stochastic processes  one considers the more general notion of {\em correlation kernels}. For any positive integer $n$, given time  tuple $\mathbf{t}_n \in T^n$, and vectors $\mathbf{a}_n =(a_1,\ldots,a_n)^{\top}  \in \mathscr{B}^n$ and $\mathbf{b}_n =(b_1,\ldots,b_n)^{\top} \in \mathscr{B}^n$, correlation kernels $w_{\mathbf{t}_n}$  are complex functions on $\mathscr{B}^n \times \mathscr{B}^n$ of the form
\begin{align}
w_{\mathbf{t}_n}(\mathbf{a}_n,\mathbf{b}_n) = \mu(j_{\mathbf{t}_n}(\mathbf{a}_n)^* j_{\mathbf{t}_n}(\mathbf{b}_n)), \label{eq:correlation-kernel}
\end{align}  
where $j_{\mathbf{t}_n}(\mathbf{a}_n)=j_{t_n}(a_n) j_{t_{n-1}}(a_{n-1}) \ldots j_{t_1}(a_1)$. For properties of correlation kernels, see \cite[Proposition 1.2]{AFL82}. If $\mathbf{a}_n=\mathbf{b}_n=\mathbf{E}_n$, where the elements $E_j$, $j=1,\ldots,n$, of $\mathbf{E}_n$ are mutually commuting events (projection operators) in $\mathscr{B}$ then $w_{\mathbf{t}_n}(j_{\mathbf{t}_n}(\mathbf{E}_n)^* j_{\mathbf{t}_n}(\mathbf{E}_n))$ gives the joint probability distribution of these events. Otherwise, it gives the probability for the events to occur in that specific time order.  For the remainder of the paper, for concreteness we take $\mathfrak{h}$ to be embeddable as a subspace of $\mathfrak{k}$, $\mathscr{A} \subseteq \mathrm{B}(\mathfrak{k})$ and $\mathscr{B}$ to be embeddable as a sub-algebra of $\mathscr{A}$.  We consider $j_t$ of the form $j_t(\cdot) = U_t^* (\cdot) U_t$, with $U_t$ a unitary operator on $\mathfrak{h}$ and define $j_t^{\star}$ via $j_t^{\star}(\cdot) = U_t (\cdot) U_t^*$. Note that by definition, $j_t^{\star}$ is the essentially the inverse of $j_t$, in the sense that $j_t \circ j_t^{\star} = j_t^{\star} \circ j_t =I$, where $I$ is an identity map. We also define $j_{t_1,t_2}(\cdot) = j_{t_2} \circ j_{t_1}^{\star}(\cdot) = U_{t_2}^*U_{t_1} (\cdot) U_{t_1}^*U_{t_2}$ and in a similar fashion define $j_{t_1,t_2}^{\star} = j_{t_1} \circ j_{t_2}^{\star}$.

\section{Multi-time correlations}
\label{sec:multi-time}

In many cases, we can consider experiments where several measurements are
made at specific times over a time interval.  

\begin{definition}
A family of experiments on a particular system is said to be of order $n$ if
each experiment consists of nontrivial measurements made at $n$ distinct
times $t_{1},\cdots ,t_{n}$. The family is said to be \textit{complete} if we make
all possible experiments exhausting everything we could measure and covering
the time interval. For $n>1$ we say that the experiments are multi-time
experiments.
\end{definition}

We insert the requirement of nontriviality to ensure that we have a hierarchy - an order $n$ experiment is not a special case of a higher order experiment. Incompatible observations being made at the same time in the same trial are precluded. Of course, the family itself can include incompatible measurements. However, we may make incompatible measurements within the same experiment so long as they are made at different times. The notion of completeness just means that we do as many measurements as possible, restricted to $n$ distinct times within the interval of interest. The aim is to be exhaustive in what is measured. 

In an order 1 experiment, each trial involves measuring the system at a single time. In an order 2 family the experimenter measures observables at different times \textit{in the same trial} and thereby obtains the two-point statistical correlations between
different quantities at different times - something not available from the
data collected in an order 1 experiment. An order $n$ experiment contains information not available from
lower order experiments.

In both classical and quantum theories, order 1 measurements lead to the notion of an ``instantaneous state''. For instance, in quantum mechanics we obtain empirically from order one experiments the expectations $\langle X(t) \rangle$: this should take the form $\mathrm{tr}(\rho_t X)$ and by varying the observable $X$ measured we determine a density matrix $\rho_t$. As such, a complete family of order 1 experiments over the time interval reveals the set of ``instantaneous states'', but this is still only partial information.
We now focus on what we can obtain empirically from a family of experiments of order $n$.

\label{prop:multi-time}
In quantum theory, the most general multi-time correlation that can be estimated from experiment are those of the form \cite{AFL82} 
\begin{eqnarray}
\langle X_{1}\left( t_{1}\right) ^{\ast }\cdots X_{n}\left( t_{n}\right)
^{\ast }Y_{n}\left( t_{n}\right) \cdots Y_{1}\left( t_{1}\right) \rangle 
\label{eq:pyramid}
\end{eqnarray}
where the times are ordered as $0\leq t_{1}<t_{2}<\cdots <t_{n}\leq T,$ and the operators are all in the Heisenberg picture at the indicated times. We refer to these as \textit{pyramidal time-ordered correlations}. The essential feature is that we have increasing times as we work in from
the outermost operators into the centre.
For instance, suppose an order $n$ experiment involves measuring a $k$th observable at time $t_k$ and that this results in an answer $\omega_k$. Let $Q_k (\omega )$ give the corresponding projection in the Heisenberg picture. 
(For yes/no experiments we have $Q_k (\text{yes}) =P_k (t_k)$ and $Q_k (\text{no}) = I-P_k (t_k) $.) From the experiment we may estimate the empirical probabilities
\begin{eqnarray*}
\lefteqn{p _n (\omega_1, \cdots , \omega _n ) }\\
&=&  \langle Q_{1}\left( \omega_{1}\right) \cdots
Q_{n}\left( \omega_{n}\right) Q_{n}\left( \omega_{n}\right) \cdots Q_{1}\left(
\omega_{1}\right) \rangle 
\end{eqnarray*}
where $\omega_k$ is the answer to the $k$ measurement at time $t_k$. From the fact that $\sum_\omega Q_k (\omega) =I$, we find
\begin{eqnarray*}
\sum_{\omega_n} p_n (\omega_1, \cdots , \omega _n ) =  p_{n-1} (\omega_1, \cdots , \omega _{n-1} ) .
\end{eqnarray*}
This may be  rephrased as follows.

\begin{proposition}
We may reduce an order $n$ experiment to an order $n-1$ experiment by ignoring the last measurement in time.
\end{proposition}

However, the projections at different times are not assumed to commute with each other. As a result, the finite dimensional distributions need not satisfy Kolmogorov's consistency conditions in any of the arguments, other than the very last one. This is a key feature of quantum theory and is the basis for results such as Bell's Theorem. As this can be misunderstood and lead to erroneous results or conclusions, such as when statistical inference methods based on the existence of joint distributions are applied to the outcomes of non-commuting sequential measurements, it will be revisited in more detail in the next section. 


\section{Sequential measurements and their subtleties}
\label{sec:sequential}

It follows from Section \ref{sec:multi-time} that  the correlation kernels $w_{\mathbf{t}_n}$ as defined in \eqref{eq:correlation-kernel} are intimately related to sequential measurements on quantum stochastic processes. In this section we explicitly illustrate the subtleties of these measurements, which can result in a sequence of random outcomes that fail the Kolmogorov consistency conditions and are therefore not classical stochastic processes. 

For simplicity of discussion, consider a discrete-valued observable $X_j \in \mathscr{B}$ (i.e., $X$ has a most a countable number of  eigenvalues) with all eigenvalues distinct.  If a measurement of $X_j$ is made at time $t_j$, the random outcome $m_j$ of the measurement will correspond to the application of a projection operator $P_{m_j} \in \mathscr{B}$ corresponding to the eigenvalue $\lambda_{m_j}$ of $X_j$ that is observed. The probability of sequentially observing the outcomes $m_1,m_2,\ldots,m_n$ at times $t_1 < t_2 <\ldots< t_n$ {\em in this order} under the evolution of the quantum stochastic process  is given by:
\begin{align*}
\lefteqn{P_{\mathbf{t}_n}(m_1,\ldots,m_n)}\\
&=\mathrm{tr}(\rho j_{t_1}(P_{m_1})^* j_{t_2}(P_{m_2})^* \cdots j_{t_n}(P_{m_n})^*j_{t_n}(P_{m_n})\\
&\quad \cdots j_{t_2}(P_{m_2}) j_{t_1}(P_{m_1})),\\ 
& =\mathrm{tr}(j_{t_n}^{\star}(P_{m_n} j_{t_{n-1},t_n}^{\star}( \cdots (P_{m_2}j_{t_1,t_2}^{\star} (P_{m_1} j_{t_1}^{\star}(\rho) P_{m_1}) P_{m_2}) \\
&\qquad \cdots )P_{m_n})).
\end{align*}
Caution is now due. In the quantum context, marginalization over any of the variables $m_j$ for any $j<n$ does not in general hold  except over the last one at time $t_n$ (a violation of the Kolmogorov consistency conditions). That is, in general
\begin{align}
\lefteqn{\sum_{m_k} P_{\mathbf{t}_n}(m_1,\ldots,m_n)} \notag\\
&\qquad \neq P_{\mathbf{t}_n\backslash t_k}(m_1,\ldots, \widehat{m_k}, \cdots,m_n),\; \forall k < n, \label{eq:marginalization}
\end{align}
where a hat ($\,\widehat{\cdot}\,$) above a variable indicates that the variable is dropped from the list of arguments. In the following, to emphasize this we give some simple but explicit examples. We mention that the general relationship \eqref{eq:marginalization} is the basis for violation of the Leggett-Garg inequalities \cite{ELN13} in sequential measurements in quantum mechanics, which is essentially a statement about the failure of the Kolmogorov consistency conditions \cite[\S 7 and Eq. (8.5)]{FL60}. For further discussions on these issues, we refer to \cite{Gough20,MSPM20}. Measurements that  satisfy $[j_{t_k}(P_{m_k}),j_{t_l}(P_{m_l})] =0 $ for all $k,l$ are referred to as {\em quantum non-demolition (QND) measurements}. QND measurements produce a classical stochastic process with a well-defined joint probability distribution for any collection of sampled points from the process.
\begin{example}
Take a qubit with Hilbert space $\mathfrak{h}=\mathbb{C}^2$ and the simple hypothetical situation where the evolution is frozen between measurements (i.e., $j_t =I$ for all $t \geq 0$). We take as basis vectors $|0\rangle =(0,1)^{\top}$ and $|1 \rangle=(1,0)^{\top}$. We analyze the sequential measurements of the Pauli X operator $X=\left[\begin{array}{cc} 0 & 1\\ 1 & 0 \end{array}\right]$ at time $t_1$ and $Z=\left[\begin{array}{cc} 1 & 0 \\ 0 & -1\end{array}\right]$ at a later time $t_2>t_1$. We consider the measurements of $X$ followed by $Z$  to show the inconsistencies that arise. Suppose that the qubit is initialised in the state $|\psi \rangle$. The probability of observing a measurement of $Z$ giving $i_1 = -1 $ followed by a measurement of $X$ giving  $i_2 =1$ is
\begin{align*}
P(\hbox{$i_1=-1$ then $i_2=1$}) &=  |\langle 0 | \psi\rangle|^2 \left| \frac{1}{\sqrt{2}}(\langle 0| - \langle 1|) | 0 \rangle \right|^2\\
&= \frac{1}{2}|\langle 0 |\psi \rangle|^2
\end{align*}
Similarly, the probability of observing a measurement of $Z$ giving a value $i_1 = 1 $ followed by a measurement of $X$ giving a value $i_2 =1$ is
 \begin{align*}
P(\hbox{$i_1=1$ then $i_2=1$}) &=  |\langle 1 | \psi \rangle|^2 \left| \frac{1}{\sqrt{2}}(|\langle 0|  - \langle 1|) | 1\rangle \right|^2\\
&= \frac{1}{2}|\langle 1 | \psi\rangle|^2
\end{align*}
So, that marginalising over $i_1$ gives:
\begin{align*}
\sum_{x={-1,1}} P(\hbox{$i_1=x$ then $i_2=1$})  &= \frac{1}{2}|\langle 0 | \psi \rangle|^2 + \frac{1}{2}|\langle 1  | \psi \rangle|^2\\
&=1.
\end{align*}
On the other hand, if we do not measure $Z$ at $t_1$ and only measure $X$ in the state $|\psi\rangle$ at time $t_2$ then we get:
 \begin{align*}
P(\hbox{$i_2=1$}) &= \frac{1}{2}|(\langle  0| - \langle 1| ) |\psi \rangle |^2.
\end{align*}
Thus we see that in general, marginalizing $i_1$ leads to inconsistency with a measurement of $Z$ only at $t_2$:
\begin{align*}
\sum_{x={-1,1}} P(\hbox{$i_1=1$ then $i_2=1$})  &\neq   \frac{1}{2}|(\langle  0| - \langle 1| ) |\psi \rangle |^2,
\end{align*}
except in the special case when $|\psi\rangle =  \frac{1}{\sqrt{2}}(|0\rangle - |1\rangle)$ so that $P(i_2=1)=1$. The reason for this is of course well understood. A measurement of $Z$ at time $t_1$ changes the quantum state and this will influence the subsequent measurement of $X$. This does not happen in a classical stochastic process, where performing a measurement does not change the probability measure underlying the process.
\end{example} 

\begin{example}
The previous example gave the sequential measurement of two non-commuting observables when the system state is frozen in between measurements. When the state is evolving, measurement of the same observable at different times may also not commute. We consider the simple qubit example again. Suppose that the qubit is initialized in the state $|\psi\rangle$ and the evolution is given by the Hamiltonian $H= \frac{1}{2}\omega Z$. We consider the measurement of $X$ at sequential times $0<t_1<t_2< \ldots < t_n$, and note that $[H,X]\neq 0$.  Let $U_t=\exp(-iHt)$. Let $i_k$ denote the outcome of measuring of $X$ at time $t_k$ and let $|i_k\rangle$ be an eigenvector of $X$ corresponding to $i_k$.  Let $P_{i_k} = |i_k\rangle \langle i_k|$. The   {\em unnormalized} state of the qubit after the $n$-th measurement is  
$$
|\psi_{t_n} \rangle = P_{i_n}U_{t_n-t_{n-1}} \cdots P_{i_{n-2}}U_{t_3-t_2}P_{i_2}U_{t_2-t_1} P_{i_1} U_{t_1}|\psi \rangle.  
$$
The probability of observing $i_k=x_k$ with $x_k\in\{-1,1\}$ is given by
\begin{align*}
\lefteqn{P(\hbox{$i_1=x_1$ then $i_2=x_2$ .... then $i_n=x_n$})} \\
&= \langle \psi_{t_n}  |\psi_{t_n} \rangle \\
&= \langle \psi | U_{t_1}^* P_{x_1} U_{t_2-t_1}^* P_{x_2} \cdots U_{t_n-t_{n-1}}^* P_{x_n} \\
&\qquad \cdots U_{t_n-t_{n-1}} \cdots  P_{x_2} U_{t_2-t_1} P_{x_1} U_{t_1} |\psi \rangle\\
&= \langle \psi | j_{t_1}(P_{x_1})j_{t_2}(P_{x_2}) \cdots j_{t_{n-1}}(P_{x_{n-1}}) j_{t_n}(P_{x_n}) \\
&\qquad \times j_{t_{n-1}}(P_{x_{n-1}}) \cdots j_{t_2}(P_{x_2}) j_{t_1}(P_{x_1}) |\psi \rangle 
\end{align*}
Let $X_t=j_t(X)$ and $Y_t=j_t(Y)$. The Heisenberg equation of motion is $\dot{X}_t = \omega Y_t$ and $\dot{Y}_t = -\omega X_t$, with initial condition $X_0=X$ and $Y_0=Y$. This has the solution $X_t = \cos(\omega t) X + \sin(\omega t)Y$ and $Y_t = -\sin(\omega t) X + \cos(\omega t)Y$ and it follows that $[X_{t_j},X_{t_k}]=\sin(\omega(t_j-t_k))[X,Y]$. For $[X_{t_j},X_{t_k}]=0$, we must have that $t_j-t_k$ must be an integer multiple of $\pi/\omega$. Since $P_{x} = \frac{1}{2}(I -\mathrm{sgn}(x)X)$, where $\mathrm{sgn}(x)$ denotes the sign of $x$,  it follows that $[j_{t_j}(P_{x_j}),j_{t_k}(P_{x_k})] =\mathrm{sgn}(x_j x_k) [X_{t_j},X_{t_k}]$. We conclude  that $[j_{t_j}(P_{x_j}),j_{t_k}(P_{x_k})] =0$ if and only if $t_j$ is of the form $t_1$ + an integer multiple of $\pi/\omega$ for all $j \geq 2$ while $t_1$ can be arbitrary. In this case, $X_{t_j}$ is either $X$ or $-X$. Also,  when the measurement is QND, given the first measurement $i_1$ at time $t_1$ (which is random) the remaining measurements $i_2,i_3,\ldots,i_n$ become deterministic for any $n >1$ since the system state can either stay at a particular eigenstate of $X$ (giving a constant sequence) or cycles in a deterministic manner between the orthogonal eigenstates of $X$. That is,  {\em the probability of observing any sequence $i_1,i_2,\ldots$ is completely determined only by the probability of observing $i_1$ alone}.

\end{example}

\section{The process tensor}
\label{sec:process-tensor}
We first motivate and introduce the notion of a discrete-time process tensor. We  start by recalling the definition of quantum operations and quantum instruments, see, e.g.,  \cite{Holevo01}.

\begin{definition}[Quantum operation]
Let $\mathfrak{h}$ be a Hilbert space. A quantum operation $O: \mathrm{B}({\mathfrak{h}}) \rightarrow  \mathrm{B}({\mathfrak{h}})$ is a linear completely positive map with the property that  $\mathrm{tr}(O\rho) \leq \mathrm{tr}(\rho)$ for all $\rho \in \mathrm{S}({\mathfrak{h}})$. \end{definition}

The set of all such quantum operations is denoted by $\mathcal{O}(\mathfrak{h})$. A special quantum operation is the ``do nothing" or identity operation $\mathrm{Id}$, defined by $\mathrm{Id}(X) = X$ for all $X \in \mathrm{B}(\mathfrak{h})$. 

\begin{definition}[Quantum instrument]
Let $\mathfrak{h}$ be a Hilbert space and $(\Omega,\mathcal{F})$ be a measurable space with $\Omega \subseteq \mathbb{R}^n$.  A quantum instrument $\mathcal{I}$ is a tuple $(\mathfrak{h},\Omega,\mathcal{F},\mathcal{M})$, where  $\mathcal{M}$ is a quantum operation valued-measure that maps elements of  the $\sigma$-algebra $\mathcal{F}$ to  $\mathcal{O}(\mathfrak{h})$, with the properties 
\begin{enumerate}
\item $\mathrm{tr}(\mathcal{M}(\Omega)\rho) = \mathrm{tr}(\rho)$ $\forall \mathcal{S}(\mathfrak{h})$.

\item For any disjoint $A_1,A_2,\ldots \in \mathcal{F}$, $\mathcal{M}(\bigcup_{k=1}^{\infty} A_k) (\rho)= \sum_{k=1}^{\infty} \mathcal{M}(A_k)(\rho)$ $\forall \rho \in \mathcal{S}(\mathfrak{h})$, where convergence is in the trace norm ($\|\cdot\|_1$) on $\mathcal{S}(\mathfrak{h})$ ($\|T\|_1=\sqrt{T^*T}$).
\end{enumerate}
The space of all such instruments is denoted by $\mathscr{I}(\mathfrak{h})$.
\end{definition}

Consider a system (labelled by a subscript $s$) with a Hilbert space $\mathfrak{h}_s$  interacting with an environment (labelled by a subscript $e$) with a Hilbert space $\mathfrak{h}_e$, and let $\mathfrak{h}_{se}=\mathfrak{h}_{s} \otimes \mathfrak{h}_{e}$. The state of the system and environment is initially in the (not necessarily factored) state $\rho_{se}$ and their joint state undergoes a joint unitary evolution between times $t_j$ and $t_{j+1}$ given by the map $\mathcal{U}_{t_j,t_{j+1}}^{se}(\cdot)=U_{t_j,t_{j+1}}^{se}  (\cdot) U_{t_j,t_{j+1}}^{se*}$, with $U^{se}_{t_j,t_{j+1}}$ unitary and $U^{se}_{t_j,t_j}=I$. At the discrete-times $0 \leq t_1 <t_2 <\ldots<t_n$ they can undergo measurements performed directly on the system, or after interacting the system (but not the environment) with some freshly prepared ancillas (i.e., ancillas that have been used are not reused on subsequent measurements) followed by measurements of compatible observables on the system and/or ancillas. This is described by  a quantum instrument $\mathcal{I}_{t_j}= (\mathfrak{h}_s,\Omega_j,\mathcal{F}_j,\mathcal{M}_j)$ at time $t_j$. Given the events $A_1,\ldots,A_{n}$ with $A_j \in \mathcal{F}_j$, the unnormalised system-environment density operator  at the time $t_n$ is given by:
\begin{equation}
\sigma_{t_n} =  \mathcal{M}_{n}(A_{n}) \circ \mathcal{U}^{se}_{t_{n-1},t_n}  \circ \cdots  \circ \mathcal{M}_{1}(A_1) \circ \mathcal{U}_{0,t_1}^{se}(\rho_{se}). 
\end{equation}

Define the map $\mathcal{T}_{\mathbf{t}_n}$ via  
\begin{align}
\lefteqn{\mathcal{T}_{\mathbf{t}_n}(\mathcal{M}_{1}(A_1),\ldots,\mathcal{M}_{n}(A_{n}))} \notag \\
&= \mathrm{tr}_{\frak{h}_e}(\mathcal{M}_{n}(A_{n})  \circ  \mathcal{U}^{se}_{t_{n-1},t_n} \circ  \cdots \circ   \mathcal{M}_{1}(A_1) \circ \mathcal{U}_{0,t_1}^{se}(\rho_{se})), \label{eq:process-tensor}
\end{align}
then we can write  
$
\sigma_{t_n} = \mathcal{T}_{\mathbf{t}_n}(\mathcal{M}_{1}(A_1), \ldots,\mathcal{M}_{n}(A_{n})).
$ 
The probability of observing the events $A_{1}, \ldots,A_{n}$ at the times $t_1,t_2,\ldots,t_n$ is given by 
$$
P_{\mathbf{t}_n}(A_{1},\ldots,A_{n})=\mathrm{tr}(\mathcal{T}_{\mathbf{t}_n}(\mathcal{M}_{1}(A_{1}),\ldots,\mathcal{M}_{n}(A_{n}))),
$$
and the system density operator $\rho_{t_n}$ at time $t_n$ is then simply the normalized version of $\mathrm{tr}_{\frak{h}_e}(\sigma_{t_n})$ given by 
$$
\rho_{t_n} = \frac{\mathrm{tr}_{\frak{h}_e}(\sigma_{t_n})}{P_{\mathbf{t}_n}(A_{1},A_{2},\ldots,A_{n})}.
$$

The map $\mathcal{T}_{\mathbf{t}_n}$ defined by \eqref{eq:process-tensor} can be viewed  as a real multilinear map from an ordered sequence of quantum operations $(O_1, O_2,\ldots,O_{n})$, corresponding to $(\mathcal{M}_1(A_1), \mathcal{M}_2(A_2),\ldots,\mathcal{M}_{n}(A_{n}))$, to an unnormalised density operator. As such, the right hand side of \eqref{eq:process-tensor} can be viewed as a real linear map on the tensor product  of quantum operations $\mathcal{O}(\mathfrak{h}_s)^{\otimes n}$, with the sequence $(O_1, \ldots,O_{n})$ being mapped to the algebraic tensor product $O_1 \otimes O_2 \otimes \cdots \otimes O_{n}$. General elements of $\mathcal{O}(\mathfrak{h}_s)^{\otimes n}$ are linear combinations of such tensor product maps and limits  thereof. They correspond to ``correlated'' measurements that involve the use of the same ancillas at different time points or the presence of correlated states between distinct ancillas at different times.   It has been shown that such maps, in the special case of finite discrete-valued measurements, have the properties  \cite{PRRFPM18,MSPM20}
\begin{enumerate}
\item[(i)] $\mathrm{tr}(\mathcal{T}_{\mathbf{t}_n}(O)) \leq 1$ $\forall O \in \mathcal{O}({\mathfrak{h}_s)}^{\otimes n}$.
\item[(ii)]  {\em Complete positivity} as a map from $\mathcal{O}(\mathfrak{h}_s)^{\otimes n}$ to $\mathrm{S}(\mathfrak{h}_{s})$. 

\item[(iii)] {\em Containment}, for any $\mathbf{s}_m \subset \mathbf{t}_n$ $(m <n)$ it holds that $\mathcal{T}_{\mathbf{s}_m}(O_{s_1} \otimes \cdots \otimes O_{s_{m}}) =\mathcal{T}_{\mathbf{t}_n}(O'_{t_1} \otimes \cdots \otimes O'_{t_{n}})$, where $O'_{t_j} = O_{t_j}$ if $t_j \in \mathbf{s}_m$, otherwise  $O_{t_j} = \mathrm{Id}$.   
\end{enumerate}

We are now ready to define the process tensor \cite{PRRFPM18,MSPM20} but stated in a more general form that allows for continuous-valued measurements. 
\begin{definition}[Process tensor]
 For  a time tuple $\mathbf{t}_n=(t_1,t_2,\ldots,t_n)$ with $0 \leq t_1 <t_2<\ldots<t_n \in T$ and a system with Hilbert space $\mathfrak{h}_s$, a process tensor $\mathcal{T}_{\mathbf{t}_n}$  is a real linear map from $\mathcal{O}(\mathfrak{h}_s)^{\otimes n}$ to  $\mathrm{S}(\mathfrak{h}_{s})$ possessing the properties (i)-(iii) stated above.   
\end{definition}

When the system Hilbert space $\mathfrak{h}_s$ is finite dimensional,  process tensors can be represented as generalized Choi many-body states and can be  cast into
a matrix-product-operator form. In this case, these representations make manipulation of process tensors convenient \cite{PRRFPM18}. 

\section{The process tensor from quantum stochastic processes}
\label{sec:pt-qsp}
We now show the relation of the process tensor to a quantum stochastic process of AFL through the correlation kernels \eqref{eq:correlation-kernel}.

Let $\mathcal{B}_s$ and $\mathcal{B}_e$ be von Neumann algebras over the system and environment Hilbert space $\mathfrak{h}_s$ and $\mathfrak{h}_e$, respectively. 
The composite space for the system is environment is the quantum probability space $ (\mathcal{B}_{se},\mu_{se})$, where $\mathcal{B}_{se} = \mathcal{B}_s \otimes \mathcal{B}_e$ and $\mu_{se}$ is a normal state on  $\mathcal{B}_{se}$. Note that the state $\mu_{se}$ is {\em not} necessarily of the factored form $\mu_s \otimes \mu_e$  for some states $\mu_s$ and $\mu_e$ on the system and environment, respectively, since  the system and environment can be initially entangled or correlated. Take the time to be $T=[0,\infty)$. We define a quantum stochastic process $\mathcal{Q}_{se}$ over $\mathcal{B}_s$ as $\mathcal{Q}_{se}=(\mathcal{B}_{se},\{j^{se}_t\}_{t\in T},\mu_{se})$.  

We now attach to $\mathcal{Q}_{se}$ an ancillary system with quantum probability space $(\mathcal{B}_{a},\mu_{a})$, where $\mathcal{B}_a$ is  a von Neumann algebra over the ancilla Hilbert space $\mathfrak{h}_a$. We then define another quantum stochastic process over $\mathcal{B}_{as}=\mathcal{B}_{a} \otimes \mathcal{B}_{s}$ as $\mathcal{Q}_{ase} = (\mathcal{B}_{ase},\{j^{ase}_t\}_{t \in T},\mu_{ase})$, where $\mathcal{B}_{ase} = \mathcal{B}_{as} \otimes \mathcal{B}_{e}$, $\mu_{ase} = \mu_{a} \otimes \mu_{se}$, and $j^{ase}_t$ acts as 
\begin{equation}
j^{ase}_t(X \otimes Y) = X \otimes j^{se}_t(Y) \in \mathcal{B}_{ase}\;, \hbox{$\forall$ $X \in \mathcal{B}_{a}$ and $Y \in \mathcal{B}_{s}$.}   
\end{equation}
That is, $j^{ase}_t$ acts non-trivially only on a factor in $\mathcal{B}_{se}$. 

 For any time tuple $\mathbf{t}_n$ with $0 \leq t_1<t_2<\ldots<t_n$, the correlation kernel $w^{ase}_{\mathbf{t}_n}$ for $\mathcal{Q}_{ase}$ is given by
 \begin{align*}
w^{ase}_{\mathbf{t}_n} (\mathbf{a}_n,\mathbf{b}_n) &= \mu_{ase}(j^{ase}_{\mathbf{t}_n}(\mathbf{a}_n)^{*}j^{ase}_{\mathbf{t}_n}(\mathbf{b}_n)),
 \end{align*} 
 where the components of $\mathbf{a}_n$ and $\mathbf{b}_n$ are operators on $\mathcal{B}_{as}$. 
 
By the polarizing identity,
\begin{equation}
X^{\ast }ZY=\frac{1}{4}\sum_{n=0}^{3}(-i)^n \left( X+e^{i\frac{n\pi}{2}}Y\right)^\ast Z\left(
X+e^{i\frac{n\pi}{2}}Y\right),  \label{eq:polarizing}
\end{equation}
it suffices to consider correlation kernels with $\mathbf{b}_n=\mathbf{a}_n$. Choose $\mathfrak{h}_a$, ancilla operators $V_{1,r_1},\ldots,V_{n,r_n}$ for $r_j=1,\ldots,\chi_j$ (with $\chi_j$ a nonnegative integer) and $j=1,\ldots,n$, and state $\mu_a$ such that $\mu_a(V_{1,r'_1}^* \cdots V_{n,r'_n}^*V_{n,r_n} \cdots V_{1,r_1}) =\mu_a(V_{1,r_1}^* \cdots V_{n,r_n}^*V_{n,r_n} \cdots V_{1,r_1}) \prod_{j=1}^n \delta_{r_j r'_j}$ for all $r_j,r'_j$, where $\delta_{jk}$ is the Kronecker delta. Let $a_j \in \mathcal{B}_{as}$ of the form $a_j = \sum_{r_j=1}^{\chi_j}    V_{j,r_j} \otimes W_{j,r_j}$, with $V_{j,r_j} \in \mathcal{B}_a$ and $W_{j,r_j} \in \mathcal{B}_s$. For this choice of $a_j$ and using the fact  $\mu_{ase}(\cdot) = \mathrm{tr}(\rho_a \otimes \rho_{se} ( \cdot))$ for some density operators $\rho_a$ on $\mathfrak{h}_a$ and $\rho_{se}$ on $\mathfrak{h}_s \otimes \mathfrak{h}_e$,  we have that,
 \begin{align*}
\lefteqn{w^{ase}_{\mathbf{t}_n}(\mathbf{a}_{n},\mathbf{a}_{n})}\\
&=\mu_{ase}(j^{ase}_{\mathbf{t}_n}(\mathbf{a}_n)^{*}j^{ase}_{\mathbf{t}_n}(\mathbf{a}_n))\\
&=\sum_{r_1=1}^{\chi_1} \cdots \sum_{r_n=1}^{\chi_n}  \mu_{ase}(j^{ase}_{t_1}(V_{1,r_1}^* \otimes W_{1,r_1}^{*}) \cdots j^{ase}_{t_n}(V_{n,r_n}^* \otimes W_{n,r_n}^{*})\\
&\quad \times j^{ase}_{t_n}(V_{n,r_n} \otimes W_{n,r_n}) \cdots  j^{ase}_{t_1}(V_{1,r_1} \otimes W_{1,r_1}) )\\
&=\sum_{r_1=1}^{\chi_1} \cdots \sum_{r_n=1}^{\chi_n}  (\mu_{a} \otimes \mu_{se})(V_{1,r_1}^* \otimes j_{t_1}^{se}(W_{1,r_1}^{*}) \cdots \\
&\quad \times V_{n,r_n}^* \otimes j^{se}_{t_n}( W_{n,r_n}^{*}) V_{n,r_n} \otimes j^{se}_{t_n}( W_{n,r_n}) \cdots V_{1,r_1} \otimes j^{se}_{t_1}( W_{1,r_1}) )\\
&=\mathrm{tr}\left( \sum_{r_1=1}^{\chi_1} \cdots \sum_{r_j=1}^{\chi_n} \alpha_{r_1,\ldots,r_n}  \mathcal{W}_{n,r_n}  \circ  j^{se \star}_{t_{n-1},t_n} \circ \cdots \circ \mathcal{W}_{2,r_2} \right.\\
&\quad \left. \circ j^{se \star}_{t_1,t_2} \circ \mathcal{W}_{1,r_1} \circ j^{se \star}_{t_1} (\rho_{se}) \vphantom{\sum_{r_1=1}^{n_1} \cdots \sum_{r_j=1}^{n_j} \alpha_{r_1,\ldots,r_n}  }\right), 
\end{align*}
where $\alpha_{r_1,\ldots,r_n} = \mu_{a}(V_{1,r_1}^* \cdots V_{n,r_n}^* V_{n,r_n} \cdots V_{1,r_1}) \geq 0$ and $\mathcal{W}_{j,r_j}: \mathcal{B}_s \rightarrow \mathcal{B}_s$ is map defined by $\mathcal{W}_{j,r_j}(\cdot) = W_{j,r_j} (\cdot) W_{j,r_j}^*$. Note that in the development above we have identified $\mathcal{W}_{j,r_j}$ with its ampliation $\mathcal{W}_{j,r_j} \otimes I$ on $\mathcal{B}_{se}$.    

Define the linear operator $\mathcal{T}^{s}_{\mathbf{t}_n}$ via, 
\begin{align*}
\lefteqn{\mathcal{T}^{s}_{\mathbf{t}_n}(\mathcal{W}_{1,r_1} \otimes \mathcal{W}_{2,r_2} \otimes \cdots \otimes \mathcal{W}_{n,r_n})}\\
 &= \mathrm{tr}_{\mathfrak{h}_e}(\mathcal{W}_{n,r_n}  \circ  j^{se\star}_{t_{n-1},t_n} \circ\cdots\circ  \mathcal{W}_{2,r_2} \circ j^{se \star}_{t_1,t_2} \circ \mathcal{W}_{1,r_1} \circ j^{se \star}_{t_1} (\rho_{se})), 
\end{align*}
and note that $\mathcal{T}^{s}_{\mathbf{t}_n}$ is defined independently of the ancilla and any of its parameters. Then we have that,
\begin{align*}
\lefteqn{w^{ase}_{\mathbf{t}_n}(\mathbf{a}_{n},\mathbf{a}_{n})}\\
&=\mathrm{tr}\left(\mathcal{T}^{s}_{\mathbf{t}_n}\left( \sum_{r_1=1}^{\chi_1} \cdots \sum_{r_n=1}^{\chi_n}  \alpha_{r_1,\ldots,r_n}\mathcal{W}_{1,r_1} \otimes \cdots \otimes \mathcal{W}_{n,r_n}\right)\right)
\end{align*}
With the complete freedom to choose $\mathfrak{h}_a$ and $\rho_a$, by taking linear combinations of $\mathcal{W}_{1,r_1} \otimes \mathcal{W}_{2,r_2} \otimes \cdots \otimes \mathcal{W}_{n,r_n}$ and   limits thereof, $\mathcal{T}^{s}_{\mathbf{t}_n}$ can be extended to a linear operator mapping from $\mathrm{CP}(\mathcal{B}_s)^{\otimes n}$ to $\mathcal{S}(\mathfrak{h}_s)$, where $\mathrm{CP}(\mathcal{B}_s)$ denotes the space of all completely positive maps on $\mathcal{B}_s$. For $\mathcal{B}_s=\mathrm{B}(\mathfrak{h}_s)$ and restricting $\mathcal{T}^{s}_{\mathbf{t}_n}$ to $\mathcal{O}(\mathfrak{h}_s)^{\otimes n}$, we  recover the process tensor from Section \ref{sec:process-tensor} but without distinguishing between correlated and uncorrelated sequential quantum operations. 


\begin{theorem}
For every correlation kernel $\omega^{se}_{\mathbf{t}_n}$ there exists a process tensor $\mathcal{T}^{s}_{\mathbf{t}_n}$ such that $\omega^{se}_{\mathbf{t}_n}(\mathbf{a}_n,\mathbf{b}_n)= \sum_{k=1}^{\ell} c_k \mathrm{tr}\left(\mathcal{T}^{s}_{\mathbf{t}_n}(\mathcal{R}_k(\mathbf{a}_n,\mathbf{b}_n)\right)$, with $\ell$ some positive integer, $c_1,\ldots,c_{\ell}$ some complex constants and $\mathcal{R}_{k}(\mathbf{a}_n,\mathbf{b}_n) \in \mathrm{CP}(\mathcal{B}_s)^{\otimes n}$,  which depends on $\mathbf{a}_n$ and $\mathbf{b}_n$, of the form
$$
\mathcal{R}_{k}(\mathbf{a}_n,\mathbf{b}_n) = \mathcal{R}_{1k} \otimes \mathcal{R}_{2k} \otimes \cdots \otimes \mathcal{R}_{nk}, 
$$
where $\mathcal{R}_{jk}(\cdot) = R_{jk}(\cdot)R_{jk}^{*}$ for some $R_{jk} \in \mathcal{B}_s$.   
\end{theorem}
\begin{proof}
By the polarizing identity \eqref{eq:polarizing} we can write $\omega^{se}_{\mathbf{t}_n}(\mathbf{a}_n,\mathbf{b}_n)=\sum_{k=1}^{\ell} c_k \omega^{se}_{\mathbf{t}_n}(\mathbf{R}_{nk},\mathbf{R}_{nk})$ for some positive integer $\ell$ and complex constants $c_k$, where   $\mathbf{R}_{nk}=(R_{1k},\ldots,R_{nk})$ for some operators $R_{jk} \in \mathcal{B}_s$.  By a similar calculation to the above, we can then write 
\begin{align*}
w^{se}_{\mathbf{t}_n}(\mathbf{a}_{n},\mathbf{b}_{n}) &= \sum_{k=1}^{\ell} c_k \omega^{se}_{\mathbf{t}_n}(\mathbf{R}_{nk},\mathbf{R}_{nk})\\
&= \sum_{k=1}^{\ell} c_k \mathrm{tr}\left(\mathcal{T}^{s}_{\mathbf{t}_n} (\mathcal{R}_{k}(\mathbf{a}_n,\mathbf{b}_n))\right),
\end{align*}
with $\mathcal{R}_{k}(\mathbf{a}_n,\mathbf{b}_n)$ is as defined in the theorem statement. 
\end{proof}

Therefore, a correlation kernel can be evaluated by evaluating a process tensor on a strict subset of $\mathcal{O}(\mathfrak{h})^{\otimes n}$. This is because the correlation kernels $w^{se}_{\mathbf{t}_n}$ capture direct measurements performed on the system, whereas the process tensor allows general quantum operations involving ancillas. 

\section{Conclusion}
\label{sec:conclu} This paper has  given a tutorial overview of the AFL theory of quantum stochastic processes, multi-time correlations and  sequential quantum measurements,  and some subtleties associated with the latter two. We then recalled the notion of a process tensor and showed its relationship to the correlation kernels of an augmented quantum stochastic process incorporating ancillas. In particular, it was shown how process tensors can be recovered from correlation kernels. 

Following from this paper, there are further connections between the AFL theory and process tensors to be studied. For instance, the notion of quantum Markov processes has already been formulated in the AFL theory (see \cite{Nurd20} for an illustration in quantum optics) and, more recently, in the process tensor framework \cite{PRRFPM18b}. The question is whether these two notions are formally equivalent, as one may expect them to be. Also, a reconstruction theorem for quantum stochastic processes based on consistency conditions on the correlation kernels has been obtained in AFL theory while a generalized extension theorem (GET) has been proposed for process tensors \cite{MSPM20} as an adaptation of the Kolmogorov extension theorem for classical stochastic processes. How the GET is connected to the AFL reconstruction  will be investigated in a future work.  

\bibliographystyle{ieeetran}
\bibliography{qsp_pc_cdc}

\end{document}